\RequirePackage{pdf14}
\documentclass[preprint,fleqn,5p,twocolumn]{elsarticle}

\usepackage{amsmath, amsfonts, amssymb, float, enumerate, bm, graphicx, mathtools, mathrsfs}

\usepackage{enumitem}
\usepackage{subfigure,tikz}
\usetikzlibrary{arrows}
\usepackage{verbatim}
\usepackage{url}
\tikzstyle{int}=[draw, fill=white!20, minimum size=2em]
\tikzstyle{init} = [pin edge={to-,thin,black}]
\usepackage{tkz-euclide,subfigure}
\usetikzlibrary{backgrounds}
\usetikzlibrary{shapes.geometric}
\restylefloat{table}
\newcounter{eg}[section]
\renewcommand{\theeg}{\arabic{section}.\arabic{eg}}
\newenvironment{examp}[1][]{\refstepcounter{eg}
\par\medskip \noindent
   \textit{Example~\theeg. #1} \rmfamily}{\hfill $\square$   \hspace{-4.5pt} \vspace{6pt}}

\colorlet{red}{black}

\newcommand{\AS}{{\mathcal A}_S }
\newcommand{\AR}{{\mathcal A}_R }

\newcommand{\el}{E_L }
\newcommand{\ef}{E_F}
\newcommand{\brs}{\mathscr B_S}
\newcommand{\brr}{\mathscr B_R}
\newcommand{\Rbb}{\mathbb R}
\newcommand{\Nsrr}{{\mathcal N}_R (s^*; \ut)}
 
\usepackage{amsmath, amssymb, bbm, xspace}
\usepackage{epsfig}
\usepackage{longtable}
\usepackage{color}
\usepackage{comment}
\usepackage{ifthen}
\newboolean{showcomments}
\setboolean{showcomments}{true}
\usepackage{courier}

\newtheorem{theorem}{Theorem}[section]

\newtheorem{lemma}[theorem]{Lemma}

\newtheorem{proposition}[theorem]{Proposition}

\newtheorem{definition}{Definition}[section]

\def\bkE{{\rm I\kern-.17em E}}
\def\bk1{{\rm 1\kern-.17em l}}
\def\bkD{{\rm I\kern-.17em D}}
\def\bkR{{\rm I\kern-.17em R}}
\def\bkP{{\rm I\kern-.17em P}}

\def\bkZ{{\bf{Z}}}

\def\bkE{{\rm I\kern-.17em E}}
\def\bk1{{\rm 1\kern-.17em l}}
\def\bkD{{\rm I\kern-.17em D}}
\def\bkR{{\rm I\kern-.17em R}}
\def\bkP{{\rm I\kern-.17em P}}

\makeatletter
\newcommand{\pushright}[1]{\ifmeasuring@#1\else\omit\hfill$\displaystyle#1$\fi\ignorespaces}
\newcommand{\pushleft}[1]{\ifmeasuring@#1\else\omit$\displaystyle#1$\hfill\fi\ignorespaces}
\makeatother

\def\bkZ{{\bf{Z}}}
\def\b12{(\beta_1,\beta_2)}

\newenvironment{example}{{\noindent \bf Example}}{\hfill $\square$\hspace{-4.5pt}\vspace{6pt}}
\newcounter{example}
\renewcommand{\theexample}{\thesection.\arabic{example}}

\newcounter{remark}
\renewcommand{\theremark}{\thesection.\arabic{remark}}

\def\Xscr{\mathcal{X}}
\def\Yscr{\mathcal{Y}}

\newlength{\noteWidth}
\long\def\notes#1{\ifinner
{\tiny #1}
\else
\marginpar{\parbox[t]{\noteWidth}{\raggedright\tiny #1}}
\fi\typeout{#1}}

 \def\notes#1{\typeout{read notes: #1}} 

\newcommand{\ut}{\mathscr{U}}

\newcommand{\ie}{i.e.\@\xspace} 

\newcommand{\Real}{\ensuremath{\mathbb{R}}}
\newcommand{\inv}{^{-1}}

\def\OPT{{\rm OPT}}

\def\Pbb{{\mathbb{P}}}
\def\Nbb{{\mathbb{N}}}

\def\Range{\mathop{{\rm range}}}

\def\spose#1{\hbox to 0pt{#1\hss}}

\def\text #1{\hbox{\quad#1\quad}}

\def\xhat{{\hat x}}

\def\nthinsp{\mskip -2   mu}

\def\superstar{^{\raise 0.5pt\hbox{$\nthinsp *$}}}
\def\SUPERSTAR{^{\raise 0.5pt\hbox{$*$}}}

\def\lamstarT {\lambda^{\raise 0.5pt\hbox{$\nthinsp *$}T}}

\def\Dscr{{\cal D}}
\def\Iscr{{\cal I}}

\def\Pscr{{\cal P}}

\def\Cscr{{\cal C}}

\def\Xscr{{\cal X}}
\def\Yscr{{\cal Y}}

\def\th{^{\rm th}}

\def\Ubar{\skew2\bar U}

\def\xhat{\skew{2.8}\widehat x}

\def\Xhat{\widehat X}

\def\non{\nonumber}

\let\forallnew\forall
\renewcommand{\forall}{\forallnew\ }
\let\forall\forallnew

		\def\bkE{{\rm I\kern-.17em E}}
		\def\bk1{{\rm 1\kern-.17em l}}
		\def\bkD{{\rm I\kern-.17em D}}
		\def\bkR{{\rm I\kern-.17em R}}
		\def\bkP{{\rm I\kern-.17em P}}
		\def\bkY{{\bf \kern-.17em Y}}
		\def\bkZ{{\bf \kern-.17em Z}}
		\def\bkC{{\bf  \kern-.17em C}}

{\begin{list}{}%
         {\setlength{\leftmargin}{#1}}%
         \item[]%
}
{\end{list}}

		\def\bsp{\begin{split}}
		\def\beq{\begin{eqnarray}}
		\def\bal{\begin{align*}}
		\def\bc{\begin{center}}
		\def\be{\begin{enumerate}}
		\def\bi{\begin{itemize}}
		\def\bs{\begin{small}}
		\def\bS{\begin{slide}}
		\def\ec{\end{center}}
		\def\ee{\end{enumerate}}
		\def\ei{\end{itemize}}
		\def\es{\end{small}}
		\def\eS{\end{slide}}
		\def\eeq{\end{eqnarray}}
		\def\eal{\end{align*}}
		\def\esp{\end{split}}
		\def\qed{ \vrule height7.5pt width7.5pt depth0pt}  

	\def\cp2problem#1#2#3#4{\fbox
		 {\begin{tabular*}{0.9\textwidth}
			{@{}l@{\extracolsep{\fill}}l@{\extracolsep{6pt}}l@{\extracolsep{\fill}}c@{}}
				#1 & & $#4 $ 
			\end{tabular*}}}

		\def\bkE{{\rm I\kern-.17em E}}
		\def\bk1{{\rm 1\kern-.17em l}}
		\def\bkD{{\rm I\kern-.17em D}}
		\def\bkR{{\rm I\kern-.17em R}}
		\def\bkP{{\rm I\kern-.17em P}}
		
		\def\bkZ{{\bf{Z}}}

\newcommand {\beeq}[1]{\begin{equation}\label{#1}}
\newcommand {\eeeq}{\end{equation}}
\newcommand {\bea}{\begin{eqnarray}}
\newcommand {\eea}{\end{eqnarray}}

\def\texitem#1{\par\smallskip\noindent\hangindent 25pt
               \hbox to 25pt {\hss #1 ~}\ignorespaces}

\def\bsp{\begin{split}}
		\def\beq{\begin{eqnarray}}
		\def\bal{\begin{align*}}
		\def\bc{\begin{center}}
		\def\be{\begin{enumerate}}
		\def\bi{\begin{itemize}}
		\def\bs{\begin{small}}
		\def\bS{\begin{slide}}
		\def\ec{\end{center}}
		\def\ee{\end{enumerate}}
		\def\ei{\end{itemize}}
		\def\es{\end{small}}
		\def\eS{\end{slide}}
		\def\eeq{\end{eqnarray}}
		\def\eal{\end{align*}}
		\def\esp{\end{split}}
		\def\qed{ \vrule height7.5pt width7.5pt depth0pt}  

\usepackage{amsmath, amssymb, xspace}
\usepackage{epsfig}
\usepackage{longtable}
\usepackage{color}
\usepackage{subfig}
\newenvironment{proof}[1][]{{\noindent \emph {Proof} #1: }}{\hfill \qed \vspace{3pt}\\ }

\def\Cscr{{\cal C}}

\makeatletter
\newcommand*{\rom}[1]{\expandafter\@slowromancap\romannumeral #1@}
\makeatother
\usepackage{natbib}

\begin{document}
\title{\LARGE \textbf{Information Revelation Through Signalling}}

\author[1]{Reema Deori}
\author[1]{Ankur A. Kulkarni}

\address[1]{{Systems and Control Engineering, Indian Institute of Technology Bombay},
{Powai}, 
{Mumbai},
	{400076}, 
	{Maharashtra},
	{India. deori.reema@iitb.ac.in, kulkarni.ankur@iitb.ac.in}}

\begin{abstract}
This paper studies a Stackelberg game wherein a sender (leader) attempts to shape the information of a less informed receiver (follower) who in turn takes an action that determines the payoff for both players. The sender chooses signals to maximize its own utility function while the receiver aims to ascertain the value of a source that is privately known to the sender. It is well known that such sender-receiver games admit a vast number of equilibria and not all signals from the sender can be relied on as truthful. Our main contribution is an exact characterization of the minimum number of distinct source symbols that can be correctly recovered by a receiver in \textit{any} equilibrium of this game; we call this quantity the \textit{informativeness} of the sender. We show that the informativeness is given by the \textit{vertex clique cover number} of a certain graph induced by the utility function, whereby it can be computed based on the utility function alone without the need to enumerate all equilibria. We find that informativeness  characterizes the existence of well-known classes of separating, pooling and semi-separating equilibria. We also compare informativeness with the amount of information obtained by the receiver when it is the leader and show that the informativeness is always greater than the latter, implying that the receiver is better off being a follower. {\color{red}Additionally, we also show that when the players play behavioral strategies, an equilibrium may not exist.}
\end{abstract}

\maketitle


\section{Introduction}
The real world has no dearth of examples of interactions between parties with asymmetric information. A player's decisions in these interactions are driven by the alignment of its utility function with that of other players  and the possibility of shaping the information of the other players in order to affect the eventual actions chosen. \textit{Signalling}  is the regime where a more informed player (or  \textit{sender} of information) attempts to shape the information of a less informed player (or the \textit{receiver})  who in turn takes an action that determines the payoff of both players. The sender cannot take actions of its own, and hence attempts to maximize its utility by sending suitable signals to the receiver.

Our interest in this paper is the quantification of the number of truthful signals revealed by the sender in the process of signalling to the receiver. The sender may choose to misreport or reshape its information in order to achieve its own benefits. Thus not all signals from the sender can be relied on as truthful. Indeed, this has long been known in the Bayesian signalling games literature; distinct classes of equilibria in which sender's signals are informative or non-informative are known through the concepts of \textit{pooling} and \textit{separating} equilibria ~\cite{sobel2009signaling}. While the non-informativeness of the sender's signals is generally understood, the precise extent of the non-informativeness and the structure of the sender's type space that is correctly revealed to the receiver is, to the best of our knowledge, not fully understood.

We seek to fill this gap in this paper. We study a particular setting of signalling in which the sender aims to maximize a general utility function, but the receiver's goal is to maximize the number of sender types correctly recovered. We ask the question -- what is the \textit{amount} of information that such a receiver can glean through signals sent by the sender? We quantify the amount of information by the minimum number of distinct source symbols or sender types that can be correctly inferred by the receiver in any equilibrium of the game. Our main contribution is a characterization of this quantity, called the \textit{informativeness} of the sender, solely in terms of the utility function of the sender. 

\subsection{Main findings}
Following is a brief outline of our model. The source or type of the sender is chosen from a finite source alphabet $\Xscr $ and is known privately to the sender. 
The sender's signalling strategy maps each source symbol to a signal and the receiver's strategy maps the signal received (noiselessly) to a recovered symbol. The sender has a utility that is a function of the recovered symbol and the source symbol. The sender chooses its signalling strategy to maximize its utility for each source symbol whereas the receiver chooses its strategy to maximize the number of symbols that are correctly recovered.

Since we are concerned with signalling, we study this problem through the Stackelberg equilibrium concept with the sender as the leader and the receiver as the follower. The sender may behave differently for different values of the source, ranging from being truthful, partially truthful or completely prevaricating, depending on its utility function. Our overarching goal is understanding the \textit{informativeness} of a sender, \ie, the minimum number of symbols correctly recovered by the  receiver in any Stackelberg equilibrium. Informativeness provides a metric by which one can measure the minimum amount of information that is revealed to the receiver in any interaction with the sender. 

In general this game admits a vast number of equilibria, each yielding a different number of recovered symbols, whereby computing the informativeness is nontrivial. 
Our first main finding is a structural characterization of any Stackelberg equilibrium strategy of the sender. We define a graph, called the \textit{strong sender graph}, on the source alphabet $ \Xscr $  and induced by the utility function $ \ut. $ We show that in any Stackelberg equilibrium strategy, the sender partitions this graph into cliques and sends a unique signal for each clique. Thus all source symbols in the same clique are mapped to the same signal. The receiver, with this limited information, can recover the source only to the granularity of these cliques.   Thanks to the structural characterization, the informativeness is given by the \textit{vertex clique cover number} of this graph. This characterization gives the informativeness as a function of only the utility function of the sender, without the need to enumerate the equilibria of the game. 

In Bayesian signalling games, one can categorize Stackelberg equilibria into three known categories namely, \textit{separating, pooling} and \textit{semi-separating} equilibria. We extend these notions in a natural manner to our game and find that the value of informativeness characterizes the existence of these equilibria. 
If the informativeness is $ |\Xscr| $, then there exist only separating equilibria, if the informativeness is unity then there exist only pooling equilibria, whereas if the informativeness is greater than unity but less than $ |\Xscr|$, there exists a semi-separating equilibrium. 

In the last part of our paper, we study a Bayesian analogue of our setting; remarkably we find that it need not admit an equilibrium. This example demonstrates some departures from the usual setting of Bayesian signalling games. Since in our setting the receiver looks to maximize the number of recovered symbols in a distribution-free sense, the receiver's objective is nonlinear in the posterior distribution. Due to this the insights and intuitions from the signalling games literature do not extend directly to our setting. Mathematically this nonlinearity also leads to discontinuities that jeopardize the existence of an equilibrium. We find these observations to be instructive and may be of independent interest.

In a related but distinct model, this setting was studied with the receiver as the leader ~\cite{vora2022shannon,vora2020zero}. There it was shown that the \textit{independence number} of another graph, called the \textit{weak sender graph}, quantifies the number of symbols that can be recovered by the receiver; this quantity is called the information extraction capacity of the sender with utility $ \ut $. Surprisingly, we show that this number never exceeds the informativeness of $ \ut.$ 
Hence the receiver always \textit{benefits} from a being a follower: the minimum amount of information recovered by the receiver when it is a follower is no less than the maximum information it recovers when it is a leader. In fact, when $ \ut $ is symmetric, there is optimization duality at play here: the information extraction capacity and the informativeness are given by \textit{primal-dual} integer programs. This  connection between a switch in the order of commitment and a primal to dual transformation is fascinating, whose implications we hope to understand better in the future.

\subsection{Related Work}
Within game theory, our model can be identified as a \textit{game with incomplete information} (in particular a \textit{sender-receiver game})  ~\cite{myerson1997game}.  
There are two broad subclasses in this setting -- \textit{screening games} and \textit{signalling games} ~\cite{rasmusen2007games}. In screening games, the uninformed player is the leader while in signalling games, the informed player is the leader, as in our model. Most of the works in such games consider a Bayesian setting; see for example ~\cite{kamenica2011bayesian} where \textit{information design} or \textit{Bayesian persuasion} is the subject. In addition to a leader-follower setting, there are settings such as the one studied by Crawford and Sobel~\cite{crawford1982strategic} where the players moved simultaneously. In the setting studied in~\cite{kamenica2011bayesian} the sender commits to its choice of policy, and the receiver responds to it. Thus, the concept of \textit{Stackelberg equilibrium}~\cite{basar1999dynamic} is used to study the sequential move game in~\cite{kamenica2011bayesian} instead of Nash equilibrium~\cite{nash1951non} as done in~\cite{crawford1982strategic}. However in all of these studies concerning the interplay of information and games, to the best of our knowledge, quantification of the information transmitted remains open. 

Having said that, strategic communication has been studied from the perspective of information theory by Akyol, Langbort and Basar in~\cite{akyol2015strategic,akyol2016information}, Treust and Tomala in~\cite{treust2019strategic,le2019persuasion} and by Vora and Kulkarni in ~\cite{vora2022shannon,vora2020zero,vora2020information}. In ~\cite{akyol2015strategic,akyol2016information}, a sequential communication setting involving quadratic distortion measures, Gaussian sources and side information is studied. Treust and Tomala studied a persuasion problem with communication constraints~\cite{le2019persuasion}, while in~\cite{treust2019strategic}  they studied a similar problem with side information with the receiver.  The main contribution is in characterizing the optimal payoff of the sender using the information theoretic capacity of the channel. Signalling games in the communication set up have also been studied in~\cite{farokhi2016estimation,nadendla2018effects,saritacs2016quadratic,saritacs2020dynamic,kazikli2021signaling,kazikli2021signaling_1,kazikli2022quadratic}. To the best of our understanding neither of these studies address the question we ask. 
Our earlier works~\cite{vora2022shannon,vora2020zero,vora2020information} study a setting similar to ours but as a  screening game, while ours is a signalling game.

There has also been a lot of recent interest in the interplay of games and graphs \cite{goyal2009connections,bramoulle2007public,pandit2018refinement,bramoulle2014strategic}. The most common model studied involves denoting the players by the nodes or the vertices of the graph where two nodes are adjacent if their utility functions are dependent on each other's action. For the graphs arising in our study, the nodes are elements from the source alphabet and the edges of the sender graphs are induced by the utility function of the sender. This, to the best of our knowledge is a novel interplay of game theory and graph theory.

The paper is organized as follows. Section \ref{sec2} formulates the problem and the Stackelberg equilibrium is characterized in Section \ref{sec4}. In Section \ref{sec5}, we introduce and characterize the informativeness of the sender and study its applications. We study behavioral strategies in Section~\ref{sec:random} and in Section \ref{sec6} we conclude the paper.

\section{Problem formulation and background}\label{sec2}

\subsection{Problem Formulation}

Let $\Xscr=\{0,1,2,\hdots  ,q-1\}$, where $q\in \Nbb$, be a set of source symbols observed by the sender and let  $x\in \Xscr$ represent any one such symbol. The  symbol $ x \in \Xscr $ can be thought of as the sender's type. The sender privately observes the symbols $ x $ and transmits a message\footnote{One may more generally consider a message space $ \Yscr $. It is easy to show that as long as $|\Yscr| \geq |\Xscr| $, then without loss of generality we can take $ \Yscr = \Xscr $.} $ y =s(x)$  to the receiver using the strategy $s: \Xscr \rightarrow \Xscr$. The receiver receives this message $ y $ noiselessly and attempts to recover the sent message using the strategy $g: \Xscr \rightarrow \Xscr$ where $g(y)=\hat {x}$ is the recovered message. 
Let $\AS=\{s| s:\Xscr \rightarrow \Xscr\}$, $\AR=\{ g| g:\Xscr \rightarrow \Xscr\}$ denote the strategy sets of the sender and receiver, respectively. 

When $x\in \Xscr$ is the symbol observed by the sender and $\hat x\in \Xscr$ is the symbol recovered by the receiver, the sender obtains a utility $\ut(\hat x, x)$, where $\ut:\Xscr \times \Xscr \to \Real $ is the utility function of the sender.  Define,
\begin{equation}\label{eq:perfectly_recovered_symbols}
	\Dscr (g,s)=\{x\in \mathcal X | g\circ s(x)=x\},
\end{equation}
as the set of perfectly recovered symbols when the sender and the receiver play $s$ and $g$, respectively. 
We posit that the sender chooses $s \in \AS$ to maximize $ \ut(g(s(x)),x) $ for each $ x $, and the receiver maximizes $|\Dscr(g,s)|$ by choosing  $g \in \AR$. Since the receiver has a deterministic objective and the sender can observe the source $ x \in \Xscr $, the probability distribution of $ x $ has no bearing on the problem so long as it has support $ \Xscr $.


We will study a game where the sender commits to a strategy first and the receiver responds it, and analyze it through the \textit{Stackelberg equilibrium} ~\cite{basar1999dynamic} with the sender as leader and the receiver as the follower.
We will also compare our results to the setting where the receiver is the leader. The latter setting has been studied in ~\cite{vora2022shannon,vora2020zero, vora2020information}. We outline both settings in the sections below.

\subsubsection{{Receiver as the follower: Deterministic strategies}} \label{sec:receiverfollower}
	All terms corresponding to the case with the receiver as the follower will be denoted with the subscript $F$ and all terms corresponding to the case with the receiver as the leader will be denoted with the
	subscript $L$.

The \textit{Stackelberg equilibrium} (S.E.) solution for the model with the receiver as the follower is defined as follows.
\begin{definition}(Stackelberg equilibrium with the receiver as follower) \label{def:SE}
$s^{*} \in \AS$ is a Stackelberg equilibrium strategy of the sender if
\begin{equation} \label{eq:S.E.01}
\begin{split}
\min\limits_{g \in\brr (s^*)}\ut(g\circ s^{*}(x),x) \geq\min \limits_{g \in \brr(s)}\ut(g\circ s(x),x),\\ \forall x\in \Xscr ,\forall s\in  \AS,
\end{split}
\end{equation}
where the \textit{best response set} of the receiver $\brr(s)$ is defined as
\begin{equation} \label{eq:BRR_01}
\brr(s)=\{ g\in \AR: |\Dscr(g,s)|\geq |\Dscr(g',s)|,\ \forall g'   \in \AR \}.
\end{equation}
Any $g^{*} \in \brr(s^{*})$  is referred to as the Stackelberg equilibrium strategy of the receiver. $ (s^*,g^*) $ is referred to as the Stackelberg equilibrium.
\end{definition}
The best response set of the receiver for a strategy $s$ of the sender is a collection of strategies $g$ that recover the maximum number of source symbols correctly.
\begin{definition}(Worst case utility and utility for correct recovery) We refer to the mapping $$x \mapsto \min \limits_{g \in \brr(s)}\ut(g\circ s(x),x)$$ as the \textit{worst case utility} for a strategy $s$ of the sender and the mapping $$x \mapsto \ut(x,x)$$ as the \textit{utility for correct recovery}. 
\end{definition}
The sender's Stackelberg equilibrium strategy is the one that maximizes the worst case utility for each $ x\in \Xscr $. 
\subsubsection{Receiver as follower: Behavioral strategies}\label{sec:intro_behav_strat}
We now introduce the setting where the players play \textit{behavioral strageies}~\cite{basar1999dynamic}. We will end with the surprising conclusion that there may not exist an equilibrium in the space of behavioral strategies. Let $X$ denote the source random variable which takes values from $\Xscr$ with prior probability distribution $p$ known to both  players.  Assume that $ p(x)>0 $ for all $ x\in \Xscr. $ We redefine the strategy sets of the sender and the receiver respectively as sets of conditional distributions as follows: $\AS=\{\pi| \pi : \Xscr \times \Xscr \longrightarrow [0,1], \sum \limits_{y\in \Xscr} \pi(y|x)=1, \forall x\in \Xscr \}$ and $\AR=\{\sigma| \sigma : \Xscr \times \Xscr \longrightarrow [0,1], \sum \limits_{x\in \Xscr}\sigma (x|y)=1, \forall y \in \Xscr \}$. Let $Y \in \Xscr$ be the signal sent by the sender and let $\Yscr(\pi):=\{y\in \Xscr| \exists x \in \Xscr \text{s.t.} \pi(y|x)>0\}$ contain the signals used by the sender when it plays $\pi \in \AS$. Let $\widehat{X}$ be the symbol mapped by the receiver and the set $\Dscr(\sigma, \pi):=\{x\in \Xscr| \Pbb(\hat{X}=x|X=x) =1\}$ be the symbols recovered by the receiver with probability one when the receiver plays $\sigma$ and the sender plays $\pi$. If $ x\in \Dscr(\sigma,\pi )$, then we say $x$ is \textit{recovered correctly}. Note that the joint distribution of $ X,Y,\Xhat $ is given by
\[ \Pbb(x,y,\xhat) = p(x)\pi(y|x)  \sigma(\xhat|y).\]
For any $\pi \in \AS$, for any $x\in \Xscr$ let $E_x = E_x(\pi):=\{y \in \Xscr : \pi(y|x)>0\}$, be the support of $ \pi(\cdot|x) $.  We assume the sender chooses $ \pi $ to maximize $ \Ubar(\pi,\sigma) $ where 
\[ \bar{U} (\pi, \sigma)=\sum \limits_{x \in \Xscr} \sum \limits_{ y \in \Yscr(\pi)} \sum \limits_{\hat{x} \in \Xscr} p(x) \pi(y|x)\sigma (\hat{x}|y)\ut(\hat{x},x), \]
whereas the receiver seeks to maximize the size of $\Dscr(\sigma,\pi).  $

		\begin{definition}(Stackelberg equilibrium for behavioral strategies)
		$\pi^{*} \in \AS$ and $\sigma^{*} \in \AR$ form a Stackelberg equilibrium strategy pair of the sender and the receiver respectively if  \begin{equation}\label{eq:S.E.}
\begin{gathered}
\pi^* \in \arg \max \limits_{\pi \in \AS} \min \limits_{\sigma \in \brr(\pi) } \bar{U}(\pi,\sigma), 
\end{gathered}\end{equation} 
and $ \sigma^* \in \brr(\pi^*)$ where $\brr(\pi)=\{\sigma^* \in \AR|\ |\Dscr(\sigma^*,\pi)|\geq |\Dscr(\sigma,\pi) |, \forall\sigma \in \AR\}$. 
		\end{definition} 

The Bayesian signalling games literature involves the sender playing behavioral strategies; in such games, the use of randomization leads to an obfuscation of the true type of the sender and thereby can result         in a benefit to the sender. This is essentially the idea behind Bayesian persuasion~\cite{kamenica2011bayesian}. In this literature the objective of the receiver is to maximize its expected utility, and it is as such linear in the posterior distribution $ \Pbb(X|Y) $. On the other hand, in our model, the number of signals recovered correctly -- which is the objective of the receiver -- is \textit{not} linear in the posterior distribution. This nonlinearity has dramatic consequences. It implies that the posterior distribution is not a sufficient statistic for the receiver's problem, and hence the intuition of strategic obfuscation from the Bayesian persuasion literature does not carry over to our setting. Moreover, the nonlinearity brings with it an inherent discontinuity, owing to which a Stackelberg equilibrium is not guaranteed to exist, even for rather simple examples.

\begin{theorem}(Existence of S.E. is not guaranteed)\label{theorem:existence_not_guranteed}
There need not exist a S.E. in behavioral strategies.
\end{theorem}
A concrete counter-example with the nonexistence of the Stackelberg equilibrium is presented in Section \ref{sec:random}. This example also illustrates the discontinuity discussed above. 

As a result of Theorem~\ref{theorem:existence_not_guranteed}, we conclude that the sender must play deterministic strategies. Our analysis is therefore limited to the case where both players play deterministic strategies.
\subsubsection{Receiver as the leader}
We now recall the model in ~\cite{vora2022shannon,vora2020zero}, where the receiver is the leader and the sender follows.
Let $ \AR' = \{g| g: \Xscr \rightarrow \Xscr \cup \{\Delta\}\}$ where $ \Delta $ is such that $ \ut(\Delta,x)=-\infty $ for all $ x\in \Xscr. $ The symbol $ \Delta $ denotes a punishment action for the receiver, and is introduced for convenience. More details on this can be found in ~\cite{vora2022shannon,vora2020zero}. 
The Stackleberg equilibrium for this setting is defined as follows.

\begin{definition}(~\cite{vora2022shannon,vora2020zero} Stackelberg equilibrium  with the receiver as the leader)
$g^{*}\in \AR'$ is a Stackelberg equilibrium strategy of the receiver if 
\begin{equation} \label{eq:S.E.02}
g^{*}\in \arg \max\limits_{g\in \AR} \min \limits_{s\in \brs(g)} |\Dscr(g,s)|,
\end{equation} where the \textit{best response set} of the sender $\brs(g)$ is given by
\begin{equation}
    \begin{split}
        \brs(g )=\{s\in \AS |\ut(g\circ s(x),x) \geq \ut(g\circ s' (x),x),\\
        \forall x\in \Xscr,\forall s' \in \AS\}.
    \end{split}
\end{equation}
Any $s^{*} \in \brs(g^{*})$ is referred to as the Stackelberg equilibrium strategy for the sender. The pair $ (s^*,g^*) $ is referred to as the Stackelberg equilibrium.
\end{definition}

The setting with the receiver as the leader defined here is an elementary version of the setting developed in ~\cite{vora2022shannon,vora2020zero}. The model in ~\cite{vora2022shannon,vora2020zero} is for $n$-blocklength sequence generated at the source where $n\geq 1$, while we focus on the case $n=1.$ We outline the results in ~\cite{vora2022shannon,vora2020zero} below.
\subsection{Prior results on the model with receiver as leader} \label{sec:2prior}
\noindent The model in ~\cite{vora2022shannon,vora2020zero} focuses on a block-communication setting with a strategic sender and a receiver communicating through a noisy or noiseless channel. $\AS' =\{s_n| s_n:\Xscr^n \rightarrow \Xscr^n\}$ and $\AR' =\{ g_n| g_n:\Xscr^n \rightarrow \Xscr^n \bigcup \{\Delta\}\}$ are the strategy set of the sender and receiver respectively where $\Delta $ is a penalty term. The source generates a sequence of symbols $X=(X_1,X_2,...X_n) \in \Xscr^n$  which is observed by the sender. After observing input $X\in \Xscr^n$, the sender encodes $X$ using the strategy $s_n \in \AS'$ such that $X$ is received by the receiver as $s_n(X)=Y \in \Xscr^n$ (output of the channel). The output of the channel is then  decoded by the receiver by using $g_n \in \AR'$. This results in recovering  $X$ as  $g_n(Y)=\hat {X}$, where $\hat X\in \Xscr^n$.

The set $\Dscr'(g_n , s_n )$ is the sequence equivalent of $\Dscr(g,s)$  which contains all those sequences $X\in \Xscr^n$ which gets recovered perfectly when $g_n$ and $s_n$ are the respective strategies of the receiver and the sender. The utility  function of the sender here is an $n$-block utility function $\ut_n :\Xscr^n \cup\{\Delta\}\times \Xscr^n \to \Rbb $  where $\ut_n(\hat{x},x)=\frac{1}{n}\sum\limits_{i=1}^n \ut(\hat{x}_i, x_i), \forall x,\hat{x} \in \Xscr^n$ and $ \ut_n(\Delta,x) =-\infty$ for all $ x\in \Xscr^n $. The goal in ~\cite{vora2022shannon,vora2020zero} is to characterize  the \textit{information extraction capacity} of the sender, defined as \[\Xi(\ut) :=\lim_{n\rightarrow \infty} \min \limits_{s_n \in \brs(g_n )} |\Dscr'(g_n ,s_n )|^{1/n},\] which is the asymptotic growth rate of the number of sequences correctly recovered, or equivalently the maximum amount of information that can be extracted from a sender acting as follower. In order to capture this quantity,  a graph called the \textit{weak sender graph} $G^n_L =(\Xscr^n, E_L)$ is introduced. The independence number of the graph $G^n_L$ is denoted by $\alpha(G^n_L)$. ~\cite{vora2022shannon,vora2020zero} showed that for any $ n $, the maximum  number of sequences that the receiver can correctly recover is  $\alpha(G^n_L)$, and that 
\begin{equation}
\Xi(\ut)= \lim_{n} {\alpha(G^n_L)^{\frac{1}{n}}}.
\end{equation}
We will use a simpler version of this result for $ n=1:$
\begin{theorem}(Maximum information theorem) \label{Theorem:Anuj}
For $n=1$, the maximum number of symbols the receiver can recover from the sender is $\alpha(G_L)$. 
\end{theorem} 
With slight abuse of notation we define the information extraction capacity  of $ \ut $ for $ n =1$ as  $\Xi(\ut):= \alpha(G_L)$.


\section{Stackelberg equilibrium of the game}\label{sec4}
This section characterizes the S.E. of the model in Section~\ref{sec:receiverfollower}. We first establish the structure of the best response of the receiver to the sender's strategy. Through this we characterize the minimum number of symbols recovered by the receiver at any Stackelberg equilibrium via the \textit{vertex clique cover number} of a suitably defined \textit{strong sender graph}. 

\subsection{Best response of the receiver}\label{sec3}
In this section we study the best response of the receiver to a strategy of the sender.
Say the sender plays a strategy $s$ which maps exactly $m$ elements to the same element. Then the receiver can recover at most one symbol out of these $m$ symbols. To describe the best response of the receiver, we introduce the concept of a \textit{receiver dilemma set}.

	Let $s\in\AS$ and $i$ be any element from $\Range(s)$. The \textit{pre-image set} of $i\in \Range(s)$ denoted by $P_i (s)$,  
	\begin{equation} \label{eq:preimage}
		P_{i} (s):= \{y\in \Xscr| s(y)=i\}.
	\end{equation}
Note that every strategy $s$ partitions the source alphabet $\Xscr$ into a set of pre-image sets. We call this partition, the  \textit{receiver dilemma set}.
\begin{definition}(Receiver dilemma set)
	The \textit{receiver dilemma set} of the strategy $s$ of the sender denoted by $\Pscr(s)$ is defined as,
	\begin{equation}\label{eq:receiversdilemmaset}
		\Pscr (s):=\{P_{i} (s):i\in \Range(s)\}.
	\end{equation}
\end{definition}

The following lemma demonstrates that for any choice $s\in\AS$ of the sender, the receiver's best response is to choose those strategies $g\in \AR$ which map every image under $s$ to one of its pre-images. This is equivalent to choosing those strategies which yield $|\Dscr(g,s)|=|\Pscr (s)|.$

\begin{lemma}\label{lemma:bestresponsetomanyone}
	If the sender plays any  $s\in \AS$ then
	\begin{enumerate}
		\item (Symbols recovered) The best response set of the receiver is characterized as follows
		\begin{equation}\label{eq:bestresponsetomanyone-01}
			\brr(s) =\{g\in \AR:|\Dscr(g,s)|=|\Pscr (s)|\}
		\end{equation} 
		\item (Best response set) Equivalently, the best response set of the receiver is given by
		\begin{equation}\label{eq:bestresponsetomanyone-02}
			\brr(s)=\{g\in \AR|\ g(i)\in P_{i}(s), \forall i\in \Range(s)\}.
		\end{equation}
		
	\end{enumerate}
\end{lemma}
\begin{proof} Consider any map $s\in \AS$. Since $ |\Range(g\circ s)|\leq |\Pscr (s)| $ for any $ g \in \AR $, we have  
	\[|\Dscr(g,s)|\leq |\Pscr (s)|,\forall g\in \AR.\]  	
	Consider $g \in \AR$ which maps each received signal $i \in \Range(s)$ to one element in $P_{i}(s)$. Clearly for this choice of $g  $, we have $ |\Dscr(g,s)| =|\Pscr (s)|. $ This proves \eqref{eq:bestresponsetomanyone-01}. 
	
	The above proof also shows that the RHS of \eqref{eq:bestresponsetomanyone-02} is included in $ \brr(s).$ 
	We now show that $ \brr(s) $ is contained in the RHS of \eqref{eq:bestresponsetomanyone-02}. To this end, we first show that no strategy of the receiver can recover more than one symbol from any pre-image set. 
	Let $ i\in \Range(s), $ and say there exists a $g$ which recovers distinct $x,y\in P_i(s)$ correctly. Then, $g\circ s(x)=x$ and $g\circ s(y)=y$, while $ s(x)=s(y)=i$. This gives $g(i)=x=y$, a contradiction. Thus for any strategy $ g \in \AR, $ the receiver can recover at most one symbol in $P_{i}(s)$ for each $ i \in \Range(s). $  Now, if $ g \in \brr(s)$ is such that  $g(i)\notin P_i (s)$ for some $ i\in \Range(s), $ then for all $ x\in P_i(s),$  we have $g\circ s(x) \neq x.$ Thus such a $ g $ can not recover any symbol in $P_i (s)$, whereby for such a $ g $ 
	$|\Dscr(g,s)|<|\Pscr (s)|$, again a contradiction to $ g\in \brr(s) $. Hence, every $ g\in \brr(s) $ satisfies the RHS of \eqref{eq:bestresponsetomanyone-02}. \end{proof}
\subsection{Characterization of the Stackelberg equilibrium}

In Lemma \ref{lemma:bestresponsetomanyone} we showed that the best response  set of the receiver $\brr(s)$ for any strategy $s$ of the sender, consists of those strategies of the receiver which map every $i\in \Range(s)$ to an element in  $P_i (s)$.




 In the following proposition, we show that the worst case utility for any S.E. strategy is identically equal to the utility for correct recovery.
\begin{proposition}(Worst case utility is the utility for correct recovery)
\label{proposition:wcu_is_correct_recovery}
If $s^*$ is a S.E. strategy of the sender then 
\begin{equation}\label{eq:wcu_for_correct_recovery01}
  \min \limits_{g \in \brr(s^{*})}\ut(g\circ s^*(x), x)=\ut(x,x), \quad \forall \ x\in \Xscr.  
\end{equation}
\end{proposition}
\begin{proof}
Consider any strategy of the sender $s\in \AS.$ Fix a $P\in \Pscr(s)$ and let $x\in P.$ In Lemma~\ref{lemma:bestresponsetomanyone}, we observed that every $ g \in \brr(s)$ maps $ s(x) $ to a $\hat{x}\in P$. Thus we have,
\begin{equation}\label{eq:invariance}
	\min\limits_{g \in \brr(s )}\ut(g \circ s (x), x)=\min\limits_{\hat{x} \in P}\ut(\hat{x}, x),\ \  \forall x \in P. 	
\end{equation} 
Since $ P $ is arbitrary this holds for all $ P \in \Pscr(s).$ 
Noting that $\min\limits_{\hat{x} \in P}\ut(\hat{x},x) \leq  \ut(x,x),\forall x\in P,$  we have
\begin{equation}\label{bound_for_min_ut}
\begin{split}
\hspace{-8mm}\min\limits_{g \in\brr (s)}\ut(g\circ s^{}(x),x) \leq  \ut(x,x), \forall x\in P, \forall P \in \Pscr(s).
\end{split}
\end{equation}
Note that if $s_\circ \in \AS$ is a one-to-one mapping then every set $P$ in $\Pscr(s_\circ)$ is a singleton set and $\brr(s_\circ)=\{s_{\circ} ^{-1}\}$. Thus, for all one-to-one mappings $s_\circ \in \AS$,
\begin{equation}\label{eq:one-one_utility}
\min \limits_{g \in \brr(s_\circ)}\ut(g\circ s_{\circ}(x),x)=\ut(x,x), \forall x\in \Xscr.
\end{equation}
 Hence, every one-to-one mapping attains the bound in \eqref{bound_for_min_ut}. And consequently, if $s=s^{*}\in \AS$ is a S.E. strategy, then equality must hold in \eqref{bound_for_min_ut}, which proves \eqref{eq:wcu_for_correct_recovery01}.
\end{proof}

In the next theorem we give a structural characterization of S.E. strategies of the sender and receiver.
\begin{theorem}(Characterization of the S.E.)
\label{Theorem:Receiver_dilemma_set}
Consider a sender with a utility $ \ut $.  $s^*\in \AS$ is a S.E. strategy of the sender if and only if 
\begin{equation}\label{eq:Receiver_dilemma_set}
\ut(x,x)\leq \ut(\hat{x},x), \quad \forall x, \hat{x} \in P ,\  \forall P \in \Pscr(s^*).
\end{equation}
Moreover, for the receiver, $ |\Dscr(g,s^*)| = |\Pscr(s^*)| $ for all $ g \in \brr(s^*). $
\end{theorem}
\begin{proof}
($\Rightarrow $)  
Let $ s^* \in \AS $ be a S.E. strategy of the sender. Then from Proposition \ref{proposition:wcu_is_correct_recovery} it is clear that  \eqref{eq:Receiver_dilemma_set} holds. 

 ($\Leftarrow$) We now show the converse. Let $s^*$ be any strategy of the sender such that \eqref{eq:Receiver_dilemma_set} holds.  Let $ P \in \Pscr(s^*) $ and $ x\in P $. 
For any $s' \in \AS, $ we have 
\begin{eqnarray}
\hspace{-7mm}\min\limits_{g\in \brr(s')}\ut(g \circ s' (x), x) &\leq & \ut(x,x), \label{eq:1} \\
&=& \min\limits_{\hat{x}\in P}\{\ut(\hat{x},x)\}, \label{eq:2}\\
&=&\min\limits_{g\in \brr(s^*)}\ut(g\circ s^*(x),x),\label{eq:3}
\end{eqnarray}
where \eqref{eq:1} follows from  \eqref{bound_for_min_ut}, \eqref{eq:2} follows from \eqref{eq:Receiver_dilemma_set} and \eqref{eq:3} follows Proposition \ref{proposition:wcu_is_correct_recovery}. Since this holds for every $ x\in P $ and every $ P\in \Pscr(s^*) $, combining the above three equations we get
\begin{equation} \label{eq:Receiver_dilemma_set01}
\begin{split}
    \min\limits_{g\in \brr(s')}\ut(g \circ s' (x), x)\leq \min\limits_{g\in \brr(s^*)}\ut(g\circ s^*(x),x),\\ \forall x\in\Xscr,\forall s' \in \AS.
    \end{split}
\end{equation}
Consequently, the strategy $s^*$ is a S.E. strategy of the sender. 

And using Lemma \ref{lemma:bestresponsetomanyone}, we can conclude that the receiver can recover  at most $ |\Dscr(g,s^*)| = |\Pscr(s^*)| $ symbols for all $ g \in \brr(s^*). $
\end{proof}

The above theorem shows that in a S.E., the sender partitions the source alphabet $\Xscr$ into subsets (the partition being the receiver dilemma set) such that for any $ x,x' $ belonging to the same subset we have $ \ut(x,x) \leq  \ut(x',x) $. It maps all symbols in a subset to the same signal. Thus, upon seeing any $ x$ from an subset $ P $ in the partition, the sender has least preference for $ x $ itself to be recovered correctly amongst all elements in $ P$. On the other hand the receiver attempts to recover the source correctly, whereby the \textit{worst case utility} $\min\limits_{x'\in P}\ut(x',x)$ for the sender becomes the \textit{utility for correct recovery} $\ut(x,x)$ for every $ x\in \Xscr $.  
	
The above theorem also shows that  any one-to-one mapping $ s^* \in \AS$ is an equilibrium strategy for the sender, since for such an $ s^* $, $ \Pscr(s^*) $ is a collection of $ q $ singleton sets and \eqref{eq:Receiver_dilemma_set} holds trivially. Moreover, if the condition in \eqref{eq:Receiver_dilemma_set} is not satisfied for any $ P\subset \Xscr $ with at least two symbols,  then   there does not exist any many-to-one S.E. strategy of the sender. Thus, in such cases even a  sender whose utility is not aligned with the objective of the receiver, \ie, a sender who prefers incorrect recovery, would have to resort to playing a one-to-one mapping at Stackelberg equilibrium. We will see an example demonstrating this  later in the paper.

There are multiple S.E. strategies $ s^* $ that induce the same receiver dilemma set. However, the above theorem shows that the utility obtained by both the sender and the receiver in equilibrium is completely characterized by the receiver dilemma set\footnote{In general, there could be multiple $s\in \AS$ with identical receiver's dilemma set $ \Pscr(s)$. But the R.H.S. of \eqref{eq:invariance} is independent of $s$. Hence, any two strategies of the sender with identical receiver's dilemma sets, have the same worst case utility.}. Thus the outcome of the game can be understood by studying these receiver dilemma sets. In the next section we build on this observation to obtain a measure of informativeness of a sender's utility function.

\section{Information revealed via signalling}\label{sec5}

\subsection{Informativeness of a utility function}
Theorem \ref{Theorem:Receiver_dilemma_set} shows that in general our game admits multiple Stackelberg equilibria.  Though the utility of the sender in equilibrium does not depend on the specific equilibrium strategy being considered (by Proposition \ref{proposition:wcu_is_correct_recovery} it is equal to the utility for correct recovery), the number of symbols recovered by the receiver is a function of the S.E. strategy of the sender. Thus the amount of information revealed to the receiver in a game depends not only on the utility function of the sender but also on the specific S.E. under consideration. We now define a measure of the amount of information that is independent on the equilibrium under consideration. 
\begin{definition}(Number of symbols recovered by the receiver in a S.E.)
	The number of symbols recovered by the receiver when the sender plays S.E. strategy $ s^* $, denoted by $\Nsrr$, is defined as $ |\Dscr(g,s^*)| $ where $ g $ is any strategy in $ \brr(s^*). $
\end{definition}
Note that the above definition is meaningful because $ |\Dscr(g,s^*)| $ is the same for each $ g \in \brr(s^*) $, by Theorem~\ref{Theorem:Receiver_dilemma_set}. Moreover, we have 
$\Nsrr = |\Pscr(s^*)|.$ We have seen that all one-to-one mappings are equilibrium strategies for the sender. From Lemma \ref{lemma:bestresponsetomanyone}, if $ s^* \in \AS$ is a one-to-one mapping then $ \brr(s^*)=\{(s^*)\inv\} $ and consequently, we get $\Nsrr=q.$ 	Thus, the maximum of $ \Nsrr $ over all equilibria $ s^* $ is $ q $ for any $ \ut $.

We define the \textit{minimum} of $ \Nsrr $ over all equilibria $ s^* $ as the \textit{informativeness} of a utility function $ \ut. $ 
\begin{definition}(Informativeness of a utility function)
	The informativeness of a utility function $\ut$ of the sender denoted by $\Iscr (\ut)$ is the optimal value of the optimization problem below.
\begin{equation}
\begin{aligned}
\Iscr(\ut):=\min_{s^*\in \AS} \quad &{\Nsrr} \\
\textrm{s.t.} \quad & s^* \text{is a sender S.E. strategy}. 
\end{aligned}
\end{equation}
\end{definition}

$ \Iscr(\ut) $ captures the minimum amount of information revealed to a receiver in such a game. This quantity is a function only of the utility function of the sender and can hence be thought of as a measure of how informative an interaction with such a sender can be for the receiver.

The number of equilibria of this game is enormous; recall that all one-to-one sender strategies constitute equilibria whereby there are at least $ q! $ equilibria in any game. Though it may appear that computing $ \Iscr(\ut) $ would require one to enumerate all equilibria of the game, we show $ \Iscr(\ut) $ can be computed from $ \ut $ alone without explicitly finding all equilibria. Our main result builds on the Theorem~\ref{Theorem:Receiver_dilemma_set} to characterize  $ \Iscr(\ut) $ in terms of a graph. Using this characterization we compute the informativeness for some examples and characterize the existence of pooling and separating equilibria of the game via the informativeness of the utility.

We recollect the following notions from graph theory. Let $G=(\Xscr,E)$ be a graph where $\Xscr$ is the set of vertices and  $E$ is the set of edges. 
A set $\Cscr \subseteq \Xscr$ is called a \textit{clique} if $(x,y)\in E, \forall x,y\in \Cscr.$ We will consider a singleton subset of $\Xscr$ a clique by definition. We call a set $\Iscr \subseteq \Xscr$ an \textit{independent  set}  if  $(x,y)\notin E,\forall x,y\in \Iscr.$ The size of the largest possible independent set of $G$ is called the \textit{independence number} of the graph and is denoted by $\alpha (G).$ A \textit{vertex clique cover} of the graph $G$ is a collection of \textit{disjoint cliques} in the graph $G$ whose union is $ \Xscr $. The \textit{vertex clique cover number} of a graph $G$ denoted by $\theta_v (G)$ is the minimum number of cliques in a vertex clique cover of the graph $G.$ 

\begin{definition}(Strong sender graph) Consider a sender with utility $ \ut $. 
The \textit{strong sender graph} $G_F=(\Xscr, E_F)$ is defined as the graph where $(x, y)\in E_F$
if $\ut (x,x) \leq  \ut (y, x)$ and $\ut (y, y) \leq \ut (x,y)$.
\end{definition}
\begin{definition}(Weak sender graph)\label{def:weak}
Consider a sender with utility $ \ut $. 
The \textit{weak sender graph} $G_L=(\Xscr, E_L)$ is defined as the graph where $(x, y)\in E_L$ if either $\ut_n (x,x) \leq  \ut_n (y, x)$ or $\ut_n (y, y) \leq \ut_n (x,y)$.
\end{definition}

 In the following theorem we show the informativeness of a utility is the \textit{vertex clique cover number} of the strong sender graph induced by the utility. 

\begin{theorem}
\label{Theorem:mimumm_symbols_recovered}
Consider a utility function of the sender $\ut.$
\begin{enumerate}
\item (Correspondence between S.E. and clique covers) If $s$ is a S.E. strategy then $\Pscr (s)$ is a  vertex clique cover of the strong sender graph $G_F$. Conversely, if  $\{P_1,...,P_l\}$ is a vertex clique cover of the strong sender graph $G_F$ then there exists a S.E. strategy $s\in \AS$ such that $\Pscr(s)=\{P_1,...,P_l\}$. 
\item (Informativeness is equal to the vertex clique cover number of the strong sender graph)
 $\Iscr(\ut)=\theta_v(G_F)$, where $ G_F $ is the strong sender graph. 
\end{enumerate}
\end{theorem}
\begin{proof} \textit{Proof of 1}:$(\Rightarrow)$: Let $s\in \AS$ be any S.E. strategy of the sender. Fix $ P \in P(s) $ and consider distinct $ x,y \in P $. 
Recalling Theorem \ref{Theorem:Receiver_dilemma_set}, we have 
$ \ut(x,x)\leq \ut(y,x) $ and $ \ut(y,y)\leq \ut(x,y) $, or equivalently, $ (x,y )\in \ef $, by the definition of strong sender graph. Since $ x,y \in P$ were arbitrary, it follows that $ P $ is a clique in $ G_F $, whereby $ \Pscr(s) $ is a clique cover of $ G_F. $

($\Leftarrow$)
Suppose $C=\{P_1,...,P_l\} $ is a vertex clique cover of $G_F$. Let $ c_1,\hdots,c_l $ be distinct symbols in $ \Xscr. $ Define the following strategy $ s^*\in \AS $ as follows
\[ s^*(x) = c_k \quad \forall \ x \in P_k, \ \forall k=1,\hdots,l.\]
Clearly, $ C = \Pscr(s^*). $ It is easy to see that $ s^* $ satisfies \eqref{eq:Receiver_dilemma_set} and hence is a S.E..

\textit{Proof of 2}: From part \textit{1} of Theorem \ref{Theorem:mimumm_symbols_recovered}, we know that a strategy $ s^* \in \AS $ is a S.E. of the sender if and only if $ \Pscr(s^*) $ is a vertex clique cover of the strong sender graph $ G_F $, and  every vertex clique cover $ C $ corresponds to a S.E. strategy whose receiver dilemma set is equal to $ C $. Moreover, $ \Nsrr =|\Pscr(s^*)|$  whereby $ \Iscr(\ut) $ is equal to the vertex clique cover number of $ G_F $.
\end{proof}

 We now study examples that demonstrate how the informativeness serves as a measure of information revelation.
\begin{examp}\label{example1} Let $\Xscr=\{0,1,2\}$  and consider a utility function $\ut_1:\Xscr \times \Xscr\longrightarrow \Rbb$ given by 
\[\ut_1=[\ut_1(i,j)]=\begin{pmatrix}
0 & 1 & 1 \\
1 & 0 & 1 \\
1 & 1 & 0 
\end{pmatrix}\]
 where $\ut_1(i,j)$ is given by the $(i,j)\th$ entry\footnote{Note that the rows and columns are numbered $ 0,\hdots,q-1 $.} in the given matrix where $ i,j \in \Xscr $. This is a case of pure misalignment of interest between the sender and the receiver. The sender here strictly prefers incorrect recovery of every symbol to correct recovery. It is intuitive that since such a sender never prefers the truth, very little information can be recovered from it. For this sender there is an edge in $G_F$ (see Fig~\ref{Figure-2}) between every distinct $x$ and $y$ in  $\Xscr$. Thus, every mapping in $ \AS $ is a S.E. strategy of the sender and consequently, we have $\Iscr(\ut_1)=1,$ thereby confirming our intuition.  
 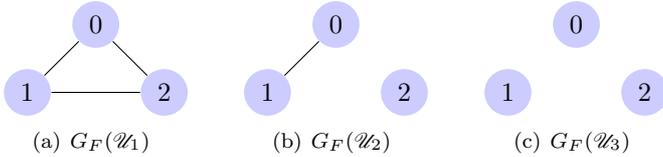
\begin{figure}[H]
\centering
\subfigure[$G_F(\ut_1 )$]{
 \begin{tikzpicture}
 [scale=.9,auto=center,every node/.style={circle,fill=blue!20}] 
   \node [ circle] (0) at ( 0, 1) {0};
   \node [ circle] (1) at (-1, 0) {1};
   \node [ circle] (2) at ( 1, 0) {2};
   \draw (0) -- (1) -- (2) -- (0);

 \end{tikzpicture}}\hfill
\subfigure[$G_F(\ut_2)$]{
 \begin{tikzpicture}
 [scale=.9,auto=center,every node/.style={circle,fill=blue!20}] 
 \node [ circle] (0) at ( 0, 1) {0};
 \node [ circle] (1) at (-1, 0) {1};
  \node [ circle] (2) at ( 1, 0) {2};
  \draw (0) -- (1);
\end{tikzpicture}}\hfill
\subfigure[$G_F(\ut_3)$]{
 \begin{tikzpicture}
 [scale=.9,auto=center,every node/.style={circle,fill=blue!20}] 
 \node [ circle] (0) at ( 0, 1) {0};
 \node [ circle] (1) at (-1, 0) {1};
  \node [ circle] (2) at ( 1, 0) {2};
  
\end{tikzpicture}}
\caption{Strong Sender Graphs}\label{Figure-2}
\end{figure}



Next we look at a slightly different sender whose utility function $\ut_2:\Xscr \times \Xscr\longrightarrow \Rbb$ is given by
\[\ut_2=[\ut_2(i,j)]=\begin{pmatrix}
0 & 1 & -1 \\
1 & 0 & 1 \\
1 & -1 & 0 
\end{pmatrix}.\]
 This sender strictly prefers that $0$ is recovered incorrectly either as $1$ or $2$. But for the symbols $1$ and $2$, it has mixed preferences. For instance it prefers that $ 1 $ is recovered as $ 0 $ as compared to $ 1 $, which in turn is preferred to $ 1 $ being recovered as $ 2 $. A similar predicament holds for the symbol $ 2 $. Here, it is intuitive that the receiver can recover the symbol $ 0 $ and at least one of $ 1 $ and $ 2 $. From the matrix it is evident that there exists only one edge, i.e., $(0,1)$ in $G_F(\ut_2)$. 
 Therefore, $\Iscr(\ut_2)=2$ for the given utility function of the sender. In other words, the receiver can recover at least two symbols out of three in any equilibrium, once again confirming our intuition.  One can verify that $\Pscr=\{\{0\},\{1\},\{2\}\}$ and $\Pscr'=\{\{0,1\},\{2\}\}$ are the only partitions which generate S.E. strategies. 

In the above two variations of the utility function of the sender, we dealt with senders whose interest misaligned with that of the receiver. Although these two cases suggest that misalignment may lead to loss of information, we find that that may not always be the case. Consider another utility function $\ut_3:\Xscr \times \Xscr\longrightarrow \Rbb$ where
\begin{equation}
\ut_3=[\ut_3(i,j)]=\begin{pmatrix}
	0 & 1 & -1 \\
	-1 & 0 & 1 \\
	1 & -1 & 0 
\end{pmatrix}. \label{eq:u3}
\end{equation}
 Observe that although for each symbol in $ \Xscr $ there is another symbol that is preferred by the sender, the exact number of symbols that can be recovered is not evident. Remarkably, we find that the receiver can recover all three symbols in any equilibrium. The utility $ \ut_3 $ is such that $G_F$ has no edges. Hence, $\Iscr(\ut_3)=3$ and indeed it is optimal for the sender to play only  one-to-one mappings. Thus, even though the sender and receiver have misaligned interests, we end up having full information revelation at every Stackelberg equilibrium. 

Why did this happen? Suppose the source symbol is $0$. Observe that although the sender prefers $0$ being observed as $2$ but for that he has to map $0$ and $2$ to the same element to confuse the receiver. After observing the sender's choice of strategy, the receiver's  best response would be a strategy which would map the sender's message to either $0$ or $2$.  Thus, the set of possible utilities of the sender is $\{\ut_3(0,0)=0,\ut_3 (2,0)=1\}$. Similarly if $2$ is the source symbol, the set of possible utilities is $\{\ut_3(0,2)=-1,\ut_3 (2,2)=0\}$. Consequently, the worst case utility the sender can get for these cases are $\ut_3(0,0)=0$ and $\ut_3(0,2)=-1$. On the other hand if the sender played any one-to-one mapping, he would have obtained utilities $\ut_3(0,0)=0$ and $\ut_3(2,2)=0$. Hence, the worst case utility obtained by the sender from a many-to-one mapping is lower than the utility of correct recovery, which is obtainable from a one-to-one mapping. Thus only one-to-one mappings are equilibrium strategies for the sender and thereby full information is revealed. 
\end{examp}


\subsection{Informativeness and pooling and separating equilibria}
Our game may be viewed as a signalling game where the  \textit{type} of the sender is the symbol observed by the sender.  Bayesian signalling games are known to admit three categories of equilibria based on whether the receiver is able to ascertain the type of the sender: these are \textit{separating, pooling} and \textit{semi-separating} equilibria \cite{sobel2009signaling}, \cite{olivella2017reputational}. We extend these notions to our game below. Our main contribution is that we can characterize the existence of these equilibria in terms of the informativeness of the utility function. 
\begin{definition}	Consider a S.E. of the game with the sender's strategy $ s^* \in \AS$. 
	\begin{itemize}
		\item \textit{Separating Equilibrium}: The S.E. is a  \textit{separating equilibrium} if for distinct $ x,x' \in \Xscr $, $ s^*(x)\neq s^*(x'). $
		\item 		\textit{Pooling Equilibrium}: The S.E. is a \textit{pooling equilibrium}  if for all $ x,x' \in \Xscr $, $ s^*(x)=s^*(x'). $
		\item \textit{Semi-separating Equilibrium}: The S.E. is said to be \textit{semi-separating equilibrium} if it is none of the above. 
	\end{itemize}
\end{definition} 

These three equilibrium classes represent qualitatively different outcomes of the game. However, they do not quantify the precise amount of information revealed. On the other hand, the informativeness notion that we have introduced is a sharper notion that serves as such a quantitative measure. Below we show that using the informativeness one can directly ascertain the existence of the equilibria, without explicitly computing all equilibria.

%
\begin{theorem}(Equivalence of informativeness and existence of different kinds of equilibria)\label{Theorem:signalingequilibria_S.E.}
	Let $ \ut $ be any utility function of the sender. 
	\begin{enumerate}
		\item For every utility function $ \ut $ of the sender, the game admits at least one  separating equilibrium.
		\item The game admits only separating equilibria if only if $ \Iscr(\ut)=q. $
		\item The game admits a pooling equilibrium if and only if $ \Iscr(\ut)=1. $
		\item The game admits a semi-separating equilibrium if and only if $ \Iscr(\ut) < q $.
	\end{enumerate}
\end{theorem}
\begin{proof}
	\textit{Proof of 1)}: Since every one-one mapping is a  S.E. strategy of the sender, hence every utility function of the sender admits at least one separating equilibrium.
	
	\textit{Proof of 2)}: From Theorem \ref{Theorem:mimumm_symbols_recovered}, we have $\Iscr(\ut)=q$ if and only if $\{ \{0\},\{1\},...,\{q-1\}\}$  is the only possible vertex clique cover for the strong sender graph of the given utility function of the sender.  Now this set of cliques $\{\{0\},\{1\},...,\{q-1\}\}$ can form a receiver dilemma set of a strategy $ s^*\in \AS $ if and only if $ s^* $ is a one-to-one mapping (using Theorem \ref{Theorem:Receiver_dilemma_set}). Therefore, the game admits only separating equilibria if and only if $\Iscr(\ut)=q.$ 
	
	\textit{Proof of 3)}: Note that in a pooling equilibrium the receiver's dilemma set $\Pscr(s^*)$ of the strategy of the sender $s^*$ has only one element $\Xscr$ such that $|\Pscr(s^*)|=1$.  Theorem~ \ref{Theorem:Receiver_dilemma_set} states that $s^*$  is a S.E. strategy iff $\Xscr$ is a clique in $ G_F $. Therefore, the game admits a pooling equilibrium if and only if $\Iscr(\ut)=1.$
	
	\textit{Proof of 4)}: From Theorem~\ref{Theorem:Receiver_dilemma_set},  $\Iscr(\ut)<q$ if and only if there exists a partition $\Pscr$ of $\Xscr$ such that $1<|\Pscr|<q$ where every element $P\in\Pscr$ is a clique in $ G_F $. Thus, from Theorem~\ref{Theorem:Receiver_dilemma_set}, we conclude that there exists a semi-separating equilibrium if and only if $ \Iscr(\ut)<q$.
\end{proof}

\subsection{Comparison with the model with the receiver as the leader}
We now compare the information recovered by the receiver in our model with that recovered when the receiver is the leader. Specifically we compare the information extraction capacity $ \Xi(\ut) $ defined in Section~\ref{sec:2prior} with the informativeness $ \Iscr(\ut).$ To motivate this comparison we consider an example below. Recall the \textit{weak sender graph} $G^L =(\Xscr, E_L)$ from Definition~\ref{def:weak}. Note that every edge in $\ef$ is also an edge in $\el$, i.e.,  $\ef \subseteq \el$. 

\begin{examp}
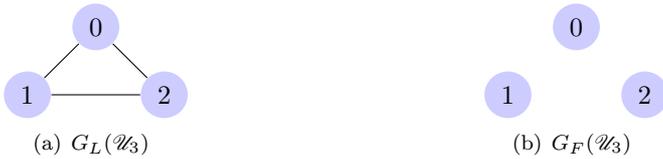
\begin{figure}[H]

\centering
\subfigure[$G_L (\ut_3 )$]{
 \begin{tikzpicture}
 [scale=.9,auto=center,every node/.style={circle,fill=blue!20}] 
   \node [ circle] (0) at ( 0, 1) {0};
    \node [ circle] (1) at (-1, 0) {1};
    \node [ circle] (2) at ( 1, 0) {2};
     \draw (0) -- (1) -- (2) -- (0);

 \end{tikzpicture}}\hfill
\subfigure[$G_F (\ut_3 )$]{
 \begin{tikzpicture}
 [scale=.9,auto=center,every node/.style={circle,fill=blue!20}] 
 \node [ circle] (0) at ( 0, 1) {0};
    \node [ circle] (1) at (-1, 0) {1};
    \node [ circle] (2) at ( 1, 0) {2};
\end{tikzpicture}}
\caption{Sender Graphs}\label{Figure-1}
\end{figure}
Recall $\ut_3 $ in \ref{example1} and recall that there exists no edge in the strong sender graph $G_F=G_F(\ut_3) $. It can be verified on the other hand that the weak sender graph $G_L=G_L(\ut_3)$ is complete (Figure \ref{Figure-1}). Clearly from Theorem  \ref{Theorem:mimumm_symbols_recovered}, it is evident that $\theta_v(G_F)=3$  while $\alpha(G_L)=1$. Therefore from Theorem \ref{Theorem:Anuj}, the minimum information recovered by the receiver for $\ut_3 $ when it is the follower is greater than the maximum information it can recover when it is the leader.

Note that only a choice of one-to-one mapping by both the players would make it possible for the receiver to recover all three symbols correctly. But if the receiver is the leader and it chooses any one-to-one mapping then its easy for the sender to get everything recovered according to its preferences. For $\ut_3$, if the receiver chooses any one-to-one mapping, the sender can easily choose such a strategy which would ensure that the $0$ gets recovered as $2$, $1$ gets recovered as $0$ and $2$ gets recovered as $1$ resulting in $|\Dscr(g,s)|=0.$ Thus, a one-to-one mapping is never optimal for the receiver when it is the leader and a choice of a many-to-one strategy $g\in \AR$ necessarily implies $|\Dscr(g,s)|<3$, for every $s\in \AS$. 
\end{examp}

We show that the above observation is true in general. For any utility function, the informativeness is no less than the  information extraction capacity of the sender.
Thus, for any utility function of the sender, the receiver benefits being a follower in the game as opposed to the setting where he is the leader as in ~\cite{vora2022shannon,vora2020zero}.

\begin{theorem}(Receiver is better off being a follower)
\label{Theorem:Receiver_better_off_as_follower}
For any utility function of the sender $\ut$,\begin{equation}\label{eq:Receiver_better_off_as_follower01}
 \Xi (\ut)\leq \Iscr(\ut).
\end{equation}

\end{theorem}
\begin{proof}
Let $\Iscr_F$ and $\Iscr_L$ be any two independent sets of the largest possible size in $G_F$ and $G_L$ respectively. From the definition of strong sender graphs and weak sender graphs, it is evident that $|\Iscr_L|\leq |\Iscr_F|$ which implies $\alpha(G_L)\leq \alpha(G_F).$ Additionally, it is well known that for any graph $G$, $\alpha(G)\leq \theta_v(G)$. 
 Thus, we can conclude that  \[\alpha(G_L)\leq \alpha(G_F)\leq \theta_v (G_F).\] And using Theorem \ref{Theorem:Anuj}, we conclude \eqref{eq:Receiver_better_off_as_follower01}.
\end{proof}

{\color{red} In the following proposition, we show that the information extraction capacity of the sender and the informativeness of the utility function of the sender are \textit{mathematical programming duals} of each other if the utility function of the sender $\ut$ is \textit{symmetric}, i.e., $\ut(x,y)=\ut(y,x), \forall x,y \in \Xscr.$
\begin{proposition}
If the utility function of the sender $\ut$ is symmetric then  $\Xi(\ut)$ and $\Iscr(\ut)$ are solutions of mathematical programming problems that are duals of each other.
\end{proposition}
\begin{proof}
Since $\ut$ is symmetric, i.e., $\ut(x,y)=\ut(y,x), \forall x,y \in \Xscr$, we have $G_F=G_L$. Let $G = (X,E)=G_F =G_L$ and let $ \Cscr $ be the set of all cliques in $ G. $ It is easy to see that $ \Xi(\ut) = \alpha(G) $ is given by the following ``primal'' problem.
\begin{equation}\label{eq_opti_primal}
\begin{aligned}
(P):\max_{x} \quad &{\sum \limits_{i \in \Xscr} x_i}\\
\textrm{s.t.} \quad & \sum \limits_{j \in c} x_j \leq 1, \quad \forall c \in \Cscr\\
  & x_j \in \{0,1\}, \quad \forall j \in \Xscr.
\end{aligned}
\end{equation}
This follows by noting that $E\subseteq \Cscr$ whereby $ \alpha(G) \geq \OPT(P)$\footnote{$\OPT(P)$ is the optimal value of the optimization problem $P$.}, whereas if $ S $ is any independent set then $ x \in \{0,1\}^{|\Xscr|}$ where $ x_i = 1 $ if and only $ i\in S $ is feasible for \eqref{eq_opti_primal} whereby $ \alpha(G)=\OPT(P). $
Now the dual of $(P)$ is given by
\begin{equation}\label{eq_opti_dual}
\begin{aligned}
(D):\min_{y} \quad &{\sum \limits_{c\in \Cscr} y_c}\\
\textrm{s.t.} \quad & \sum \limits_{c:i \in c} y_c \geq 1 , \quad \forall i \in \Xscr\\
  & y_c \in \{0,1\}, \quad \forall c \in \Cscr.\\
\end{aligned}
\end{equation}
Thus, if $y^*$ is a solution of \eqref{eq_opti_dual}, then 
 $\sum \limits_{c: i \in c} y_{c } ^* \geq 1$ for all $ i\in \Xscr $, implying that every symbol is covered by at least one clique $ c $ such that $ y_c =1 $. This is exactly the covering requirement, whereby $ \OPT(D)=\theta_v(G)=\Iscr(\ut). $
\end{proof}

In this light, when $ \ut $ is symmetric, Theorem~\ref{Theorem:Receiver_better_off_as_follower} is essentially a statement of weak duality. We find this relation fascinating and unusual -- there is no precedence known to us of a switch in the order commitment in a leader-follower setting leading to a primal to dual transformation. Understanding the implications of the above result more deeply is a task for future work.

Our study so far has restricted players to play only deterministic strategies. The next section is dedicated to investigating the case where the players are free to use behavioral  strategies.
\section{Behavioral strategies} \label{sec:random}

The observations in the previous sections have been made by confining the sender's choice to only deterministic (pure) strategies. This in turn resulted in deterministic strategies for the receiver as well. This section is devoted to studying the \textit{Stackelberg equilibrium} when the players are free to use \textit{behavioral strategies} as defined in Section \ref{sec:intro_behav_strat}. As a conclusion we will prove Theorem~\ref{theorem:existence_not_guranteed}, \ie, we present an example where a S.E. does not exist.

To begin, we make few observations about the set $ \Dscr(\sigma,\pi)$ defined in Section \ref{sec:intro_behav_strat} in the following results. The first is about the  well-posedness of our problem. The receiver always recovers at least one symbol correctly for any strategy of the sender. 
	
	\begin{lemma}\label{lemma_BR-02}
		(Receiver can always recover at least one symbol correctly) For every $\pi \in \AS$, there exists a $\sigma \in \AR$ such that $|\Dscr(\sigma , \pi)|=1.$
	\end{lemma}
	\begin{proof}
Fix a $x\in \Xscr$ and let $ \pi \in \AS $ be arbitrary. Take $\sigma \in \AR$ such that $\sigma (x|y)=1, \forall y \in \Xscr.$ Thus, $x\in \Dscr(\sigma, \pi)$ since
		\begin{equation}
		\Pbb (\widehat{X}=x|X=x)  =\sum \limits_{y \in \Xscr} \pi(y|x)\sigma(x|y) =1.
		\end{equation}
%
		
	\end{proof}
	Thus, given any strategy of the sender, the receiver can always recover at least one element correctly with probability one. 
	
	\begin{lemma} \label{lemma:sigma=1}
		If $x\in \Dscr(\sigma,\pi)$ then $\sigma(x|y)=1, \forall y \in E_x.  $
	\end{lemma}
	\begin{proof}
		Assume that there exists a $y'\in \Yscr(\pi)$ such that $0\leq \sigma(x|y')<1,$ where $\pi(y'|x)>0.$  Then, $\Pbb (\hat{X}=x|X=x)= \sum \limits_{y \in \Yscr(\pi)} \pi(y|x) \sigma(x|y)<1.$ Consequently, $x \notin \Dscr(\sigma, \pi)$ which is a contradiction. Therefore, if $x\in \Dscr(\sigma,\pi)$ then $\sigma(x|y)=1, \forall y\in E_x$.
	\end{proof}

\begin{proposition}\label{prop:E_x}
(Condition on correct recovery of two distinct symbols) For  $\sigma \in \AR$ and $\pi \in \AS$, if distinct $x,x'\in \Dscr(\sigma,\pi)$ then $E_x \cap E_{x'} =\emptyset.$
\end{proposition}
\begin{proof}
Assume $E_x \cap E_{x'} \neq \emptyset$ and let $y \in E_x \cap E_{x'}$. Note that  $x$ and $x' \in \Dscr(\sigma, \pi)$. Therefore from Lemma \ref{lemma_BR-02}, we have $\sigma(x|y)=1$ and $\sigma(x'|y)=1$ which is a contradiction since $x\neq x'$. Thus, $E_x \cap E_{x'} =\emptyset$ if $x,x' \in \Dscr(\sigma,\pi)$ where $ x\neq x' $.
\end{proof}
Therefore, if two distinct symbols belong to the set $\Dscr(\sigma, \pi)$, then they must produce distinct signals under the strategy $\pi$ of the sender. 

We now come to our main observation in this section -- that an equilibrium need not exist for this game. We show this through the following example.
	\begin{examp}\label{example_behavioral}
		Recall  $\ut_3 $ in Example \ref{example1}: \\
		\[[\ut]:=[\ut_3]=\begin{pmatrix} 
0 &1 & -1\\
-1 & 0 & 1\\
1 & -1 & 0
\end{pmatrix},\] where $\Xscr=\{0,1,2\}$. $y_1,y_2,y_3$ are distinct elements of $\Yscr(\pi)$  and $i,j,k $ are distinct symbols of $\Xscr$. $X$ is uniformly distributed on $\Xscr$. Let the prior $ p $ be uniform on $ \Xscr. $

Note that for the given  $\ut$, $\ut(x,y)=-\ut(y,x), \forall x,y \in \Xscr=\{0,1,2\}$. Additionally, if $x,y,z \in \Xscr$ are distinct symbols in $\Xscr$ then $\ut(x,y)=-\ut(x,z)$ and $\ut(x,y)=-\ut(z,y)$. 

We categorize the $\pi$'s in $\AS$ into three different classes on the basis of size of $\Dscr(\sigma, \pi),$ for $\sigma \in \brr(\pi)$. We use the notation `$ \pi  \in \bullet$' to denote that $ \pi $ lies in class `$ \bullet $'. Associated figures depict $ \pi $; for example in Figure \ref{Figure-CAI}, an edge from $ i $ to say $ y_1 $ says that $ y_1\in E_i $. The three classes are as follows:\\
\\
\textbf{(A)} \quad $|\Dscr(\sigma, \pi)|=3, \forall \sigma \in \brr(\pi)$:\\
This is possible when the sender plays any strategy $\pi \in \AS$ where distinct signals are used for distinct symbols, \ie, $\Yscr(\pi)=\{y_1,y_2,y_3\}=\Xscr, E_i=\{y_1\},E_j=\{y_2\},$ and $E_k=\{y_3\}$ where $i,j,k$ are distinct elements in $\Xscr$ and $y_1,y_2,y_3$ are distinct elements in $\Yscr(\pi)$.
Thus, $\brr(\pi)=\{\sigma\}$ such that $\Dscr(\sigma,\pi)=\Xscr.$ Hence, $\bar{U}(\pi,\sigma)=\sum \limits_{x\in \Xscr} p(x)\ut(x,x)= \frac{1}{3}[0]=0$ and  consequently, 
\begin{equation}\label{eq:A}
	\min \limits_{\sigma \in \brr(\pi)} \bar{U}(\pi, \sigma)=0, \quad \forall \pi\in {\rm A}.
\end{equation}
\\
\textbf{(B)} \quad $|\Dscr(\sigma, \pi)|=1,  \forall \sigma \in \brr(\pi)$:\\
Fix $ k\in \Xscr $ and let $\sigma' (k|y)=1,\forall y \in \Yscr(\pi)$. For every $\pi$ in class (B), $\sigma' \in \brr(\pi)$ since $\Dscr(\sigma' , \pi)=\{k\}$. Thus for every $\pi$ in class (B),
  \begin{equation}
  \begin{gathered}
  \bar{U}(\pi, \sigma')
  =\frac{1}{3}[\ut(k,i)+ \ut(k,j)]=0.
  \end{gathered}
  \end{equation}Consequently, \begin{equation}\label{eq:B}
  \min \limits_{\sigma \in \brr(\pi)} \bar{U}(\pi, \sigma)= 0,\quad \forall \pi \in {\rm B.}
\end{equation}
 
\textbf{(C)} \quad $|\Dscr(\sigma, \pi)|=2,  \forall \sigma \in \brr(\pi)$:\\
This class can be subdivided into two subclasses based on the the elements in $\Dscr(\sigma, \pi)$: either $ \Dscr(\sigma,\pi) $ is the same for all $ \sigma \in \brr(\pi) $ or $ \Dscr(\sigma,\pi) $ is one of two distinct sets, $ \{i,j\} $ and $ \{i,k\} $ where $ i,j,k \in \Xscr $, for each $ \sigma \in \brr(\pi) $. Note that a third case where $ \Dscr(\sigma,\pi) $ can take be one of \textit{three} sets $ \{i,j\} $, $ \{j,k\} $ and $ \{i,k\} $ is not possible due to Proposition~\ref{prop:E_x}.\\
\\
\textbf{C.(a)}: For distinct $i,j,k\in \Xscr$ there exist  $\sigma_{ij}, \sigma_{ik} \in \brr(\pi)$ such that $\Dscr(\sigma_{ij}, \pi)=\{i,j\}$ and $\Dscr(\sigma_{ik}, \pi)=\{i,k\}$:\\
\\
We have the following three cases in this class: \\
\\
\indent \textbf{i)} $\Yscr(\pi)=\{y_1 ,y_2\}$ for some distinct $y_1, y_2 \in \Xscr$: Under this case we have $E_i=\{y_1\}, E_j =\{y_2\}$ and $E_k=\{y_2\}$ (see Figure \ref{Figure-CAI}). Thus, for every $\pi$ in this class, $\brr(\pi)=\{\sigma \in \AR| \Dscr(\sigma,\pi)=\{i,j\} \text{or} \{i,k\}\}$. Let $\sigma_{ij},\sigma_{ik}\in \brr(\pi)$ be such that $\Dscr(\sigma_{ij}, \pi)=\{i,j\}$ and $\Dscr(\sigma_{ik}, \pi)=\{i,k\}$. Accordingly, $ \bar{U}(\pi, \sigma_{ij})=\frac{1}{3} [\pi(y_2|k)\ut(j,k)]=\frac{1}{3}\ut(j,k)$ and  $ \bar{U}(\pi, \sigma_{ik})=\frac{1}{3}\ut(k,j)$. Since $\ut(j,k)=-\ut(k,j)$ and  $|\ut(j,k)|=1$, therefore we can conclude that \begin{equation}\label{eq:cai}
	\min \limits_{\sigma \in \brr(\pi)} \bar{U}(\pi, \sigma)=-\frac{1}{3}, \forall \pi \in {\rm C.(a){\rm -}i)}.
\end{equation}

\begin{figure}[H]
\centering
\begin{tikzpicture}
  [scale=.9,auto=center,every node/.style={circle,fill=blue!20}] 
    \node (a1) at (1,1) {$k$};  
  \node (a2) at (1,0)  {$j$}; 
  \node (a3) at (1,-1)  {$i$};  
   \node (a4) at (3,0) {$y_2$};
    \node (a5) at (3,-1) {$y_1$}; 
  
  \draw (a1) -- (a4); 
  \draw (a2) -- (a4); 
  \draw (a3) -- (a5);

 \end{tikzpicture}
\caption{ C.(a) - i)}
\label{Figure-CAI}
\end{figure}
\textbf{ii)} $\Yscr(\pi)=\{y_1 ,y_2, y_3\}$ and $E_i=\{y_1,y_2\}$ for some distinct $y_1, y_2, y_3 \in \Xscr$: In this particular case, $E_j =E_k =\{y_3\}$ are the only possible values of $E_j$ and $E_k$ respectively, thanks to Proposition~\ref{prop:E_x} (see Figure \ref{Figure-004}). Therefore, for every $\pi$ in this class, $\brr(\pi)=\{\sigma \in \AR| \Dscr(\sigma,\pi)=\{i,j\} \text{or} \{i,k\}\}$ and accordingly we get  $\min \limits_{\sigma \in \brr(\pi)} \bar{U}(\pi, \sigma)=-\frac{1}{3}, \forall \pi$ in class C.(a)-ii).
\begin{figure}[H]
\centering
\begin{tikzpicture}
  [scale=.9,auto=center,every node/.style={circle,fill=blue!20}] 
   
  \node (a1) at (1,1) {$k$};  
  \node (a2) at (1,0)  {$j$}; 
  \node (a3) at (1,-1)  {$i$};  
   \node (a4) at (3,1) {$y_3$};
    \node (a5) at (3,0) {$y_2$};
    \node (a6) at (3,-1) {$y_1$}; 
  
  \draw (a1) -- (a4); 
  \draw (a2) -- (a4);  
  \draw (a3) -- (a5); 
  \draw (a3) -- (a6);

 \end{tikzpicture}
\caption{ C.(a) - ii)}
\label{Figure-004}
\end{figure} 
\textbf{iii)} $\Yscr(\pi)=\{y_1 ,y_2, y_3\}$ for some distinct $y_1, y_2, y_3 \in \Xscr$ and $E_i=\{y_1\}$:  
This case has two sub-cases:\\
\\
\indent \textbf{1)} $E_j=\{y_2\}$ and $E_k =\{y_2 , y_3\}$ (Figure \ref{Figure-00} (a)): Note that for every $\pi$ in this class, $\brr(\pi)=\{\sigma \in \AR| \Dscr(\sigma,\pi)=\{i,j\} \text{or} \{i,k\}\}$. Let $\sigma_{ij},\sigma_{ik} \in 
\AR$ such that $\sigma_{ij}(i|y_1)=\sigma_{ij}(j|y_2)=\sigma_{ij}(j|y_3)=1$ and $\sigma_{ik}(i|y_1)=\sigma_{ik}(k|y_2)=\sigma_{ik}(k|y_3)=1$. Thus, $\Dscr(\sigma_{ij},\pi)=\{i,j\}$, $\Dscr(\sigma_{ik},\pi)=\{i,k\}$ and accordingly, $\sigma_{ij},\sigma_{ik} \in \brr(\pi)$. Hence, $\bar{U}(\pi, \sigma_{ij})=\frac{1}{3} [\ut(j,k)]$ and $\bar{U}(\pi, \sigma_{ik})=\frac{1}{3} [\ut(k,j)]=-\bar{U}(\pi, \sigma_{ij})$, resulting in $\min \limits_{\sigma \in \brr(\pi)} \bar{U}(\pi, \sigma)= -\frac{1}{3}, \forall \pi$ in class C.(a)-iii)-1). \\
\indent \textbf{2)} $E_j=\{y_2,y_3\}$ and $E_k =\{y_2 , y_3\}$ (Figure \ref{Figure-00}(b)): For every $\pi$ in this class, $\brr(\pi)=\{\sigma_{ij},\sigma_{ik}\},$ where $\Dscr(\sigma_{ij},\pi)=\{i,j\}$ and $\Dscr(\sigma_{ik},\pi)=\{i,k\}$. Thus,$\min \limits_{\sigma \in \brr(\pi)} \bar{U}(\pi, \sigma)=-\frac{1}{3}, \forall \pi$ in class C.(a)-iii)-2).

Therefore, from the study of the three cases in C.(a), we conclude that \begin{equation}\label{eq:Ca}
	\min \limits_{\sigma \in \brr(\pi)} \bar{U}(\pi, \sigma)=-\frac{1}{3} <0, \forall \pi \in {\rm C.(a).}
\end{equation}
\begin{figure}[H]\label{pic-8}
\makeatletter 
\renewcommand{\thefigure}{\@arabic\c@figure}
\makeatother
\centering

\subfigure[C.(a)-iii)-1)]{
 \begin{tikzpicture}
 [scale=.9,auto=center,every node/.style={circle,fill=blue!20}] 
   \node (a1) at (1,1) {$k$};  
  \node (a2) at (1,0)  {$j$}; 
  \node (a3) at (1,-1)  {$i$};  
    \node (a6) at (3,1) {$y_2$};
    \node (a5) at (3,0) {$y_3$};
    \node (a4) at (3,-1) {$y_1$}; 
  
  \draw (a3) -- (a4); 
  \draw (a2) -- (a6);  
  
  \draw (a1) -- (a5); 
  \draw (a1) -- (a6);  
\end{tikzpicture}}
\subfigure[C.(a)-iii)-2)]{
 \begin{tikzpicture}
  [scale=.9,auto=center,every node/.style={circle,fill=blue!20}] 
  \node (a1) at (1,1) {$k$};  
  \node (a2) at (1,0)  {$j$}; 
  \node (a3) at (1,-1)  {$i$};  
  \node (a4) at (3,1) {$y_3$};
    \node (a5) at (3,0) {$y_2$};
    \node (a6) at (3,-1) {$y_1$}; 
  
  \draw (a1) -- (a4); 
  \draw (a1) -- (a5);  
  \draw (a2) -- (a5);
   \draw (a2) -- (a4); 
  \draw (a3) -- (a6);
  
\end{tikzpicture}}
\caption{C.(a)-iii)} \label{Figure-00}
\end{figure}

\textbf{C.(b)}: There exists distinct $i,j\in\Xscr$ such that $\Dscr(\sigma, \pi)=\{i,j\},  \forall \sigma \in \brr(\pi)$:\\
\\
We have the following two sub-classes in this class:\\
\\
\textbf{i)}\quad  $\Yscr(\pi)=\{y_1 , y_2\}$ for some distinct $y_1 , y_2 \in \Xscr$: For this particular class, we have $E_i=\{y_1\}, E_j =\{y_2\}$ (using Proposition \ref{prop:E_x}). Since $E_k$ can never be empty, thus $E_k=\{y_1\}$ or $\{y_2\}$ or $\{y_1,y_2\}.$ Say $E_k=\{y_1\}$ then there exists a $\sigma' \in \AR$ such that $\Dscr(\sigma', \pi)=\{k,j\}$. This makes $\sigma' \in \brr(\pi)$, which violates the property of this class. Similarly, we conclude that $E_k \neq \{y_2\}$ which makes $E_k =\{y_1 , y_2\}$ (Figure \ref{Figure-002}). Therefore, $\brr(\pi)=\{\sigma_{ij}\in \AR|\Dscr(\sigma_{ij}, \pi)=\{i,j\}\}$. Thus, 
\begin{align}\label{eq:c_a_i}
\hspace{-.5cm}\min \limits_{\sigma \in \brr(\pi)} \bar{U}(\pi,\sigma)&=\frac{1}{3} [\ut(i,i)+\pi(y_1|k)\ut(i,k)\non \\& + \pi(y_2|k)\ut(j,k) + \ut(j,j)]\non \\ =\frac{1}{3}&[\pi(y_1|k)-\pi(y_2|k)]\ut(i,k) \leq \frac{1}{3},
\end{align}
where we have used that $\ut(j,k)=-\ut(i,k). $ Using $i=0,j=2,k=1$, we can conclude that
%
\begin{equation}\label{eq:supp-3}
\sup \limits_{\pi\in {\rm C.(b)-i})} \min \limits_{\sigma \in \brr(\pi)}\bar{U}(\pi, \sigma) = \frac{1}{3}.
\end{equation}
\begin{figure}[H]
\centering
\begin{tikzpicture}
  [scale=.9,auto=center,every node/.style={circle,fill=blue!20}] 
  \node (a1) at (1,1) {$k$};  
  \node (a2) at (1,0)  {$j$}; 
  \node (a3) at (1,-1)  {$i$};  
   \node (a4) at (3,0) {$y_2$};
    \node (a5) at (3,-1) {$y_1$}; 
  
  \draw (a1) -- (a5); 
  \draw (a1) -- (a4); 
  \draw (a2) -- (a4);  
  \draw (a3) -- (a5); 
  
 \end{tikzpicture}
\caption{ C.(b) - i)}
\label{Figure-002}
\end{figure}
\noindent \textbf{ii)}\quad $\Yscr(\pi)=\{y_1 , y_2, y_3\}$ for some distinct $y_1 , y_2, y_3 \in \Xscr$: This class has two-sub cases:\\
\\
\textbf{1)}\quad  $E_i=\{y_1\}, E_j =\{y_2\}$: The only way $i$ and $j$ are recovered correctly with probability one in every best response strategy is if $E_k=\{y_1, y_2, y_3\}$ (see Figure \ref{Figure-001}). Thus, for any $ \sigma \in \brr(\pi), $
\begin{align}
\bar{U}(\pi, \sigma)&=\frac{1}{3}[-\pi(y_1|k) +\pi(y_2|k)+\pi(y_3|k)(-\sigma (i|y_3) \non \\& +\sigma (j|y_3))]\ut(j,k)\leq \frac{1}{3}
\end{align}

If $i=2,j=0,k=1$, then
\begin{equation}\label{eq:cbii1}
	\min \limits_{\sigma \in \brr(\pi)}\bar{U}(\pi, \sigma)=\frac{1}{3}[-\pi(y_1|k) +\pi(y_2|k) - \pi(y_3|k)]. \non 
\end{equation} Therefore,
\begin{equation}\label{eq:supp-2}
\sup \limits_{\pi\in {\rm C.(b)-ii)-1)}}\min \limits_{\sigma \in \brr(\pi)} \bar{U}(\pi, \sigma)= \frac{1}{3}.
\end{equation}
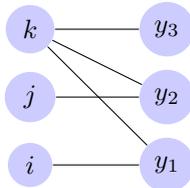
\begin{figure}[H]
\centering
\begin{tikzpicture}
  [scale=.9,auto=center,every node/.style={circle,fill=blue!20}] 
   \node (a1) at (1,1) {$k$};  
  \node (a2) at (1,0)  {$j$}; 
  \node (a3) at (1,-1)  {$i$};  
   \node (a4) at (3,1) {$y_3$};
    \node (a5) at (3,0) {$y_2$};
    \node (a6) at (3,-1) {$y_1$}; 
  
  \draw (a3) -- (a6); 
  \draw (a2) -- (a5);  
  \draw (a1) -- (a4); 
  \draw (a1) -- (a5); 
  \draw (a1) -- (a6);  
  
 \end{tikzpicture}
\caption{ C.(b) - ii) - 1)}
\label{Figure-001}
\end{figure}
\noindent\textbf{2)} \quad $E_i=\{y_1\}, E_j =\{y_2,y_3\}$:
For every $\pi$ in this subclass, we have $\brr(\pi)=\{\sigma_{ij}\}$ where $ \sigma_{ij}(i|y_1)=1,\sigma_{ij}(j|y_2)=1,\sigma_{ij}(j|y_3)=1 $. 
Accordingly, we have 

\noindent\textbf{($\alpha$)} If $E_k =\{y_1, y_2\}$ (see Figure \ref{Figure-0} (a)): 
\begin{equation}\label{alpha_eq}
\begin{gathered}
\min \limits_{\sigma \in \brr(\pi)} \bar{U}(\pi, \sigma) =\frac{1}{3}[-\pi(y_1 |k)+ \pi(y_2 |k)]\ut(j,k).\\
\end{gathered}
\end{equation}
and 

\noindent \textbf{($\beta$)} If $E_k =\{y_1, y_2, y_3 \}$ (see Figure \ref{Figure-0} (b)): 
\begin{equation}\label{beta_eq}
\min \limits_{\sigma \in \brr(\pi)} \bar{U}(\pi, \sigma)= \frac{1}{3}[-\pi(y_1 |k)+ \pi(y_2 |k) + \pi(y_3 |k)]\ut(j,k).
\end{equation}
Using $i=2,j=0,k=1$ in \eqref{alpha_eq} and \eqref{beta_eq}, we conclude that
\begin{equation}\label{eq:supp-1}
\sup \limits_{\pi\in {\rm C.(b)-ii)-2)}} \min \limits_{\sigma \in \brr(\pi)} \bar{U}(\pi,\sigma) = \frac{1}{3}.
\end{equation}

\begin{figure}[H]

\centering

\subfigure[C.(b)-ii)-2-$\alpha$]{
 \begin{tikzpicture}
 [scale=.9,auto=center,every node/.style={circle,fill=blue!20}] 
   \node (a1) at (1,1) {$k$};  
  \node (a2) at (1,0)  {$j$}; 
  \node (a3) at (1,-1)  {$i$};  
    \node (a4) at (3,1) {$y_1$};
    \node (a5) at (3,0) {$y_2$};
    \node (a6) at (3,-1) {$y_3$}; 
  
  \draw (a1) -- (a4); 
  \draw (a1) -- (a5);  
  \draw (a3) -- (a4); 
  \draw (a2) -- (a5); 
  \draw (a2) -- (a6);  
\end{tikzpicture}}
\subfigure[C.(b)-ii)-2-$\beta$]{
 \begin{tikzpicture}
  [scale=.9,auto=center,every node/.style={circle,fill=blue!20}] 
  \node (a1) at (1,1) {$k$};  
  \node (a2) at (1,0)  {$j$}; 
  \node (a3) at (1,-1)  {$i$};  
  \node (a4) at (3,1) {$y_3$};
    \node (a5) at (3,0) {$y_2$};
    \node (a6) at (3,-1) {$y_1$}; 
  
  \draw (a3) -- (a6); 
  \draw (a2) -- (a4);  
  \draw (a2) -- (a5);
   \draw (a1) -- (a4); 
  \draw (a1) -- (a5); 
  \draw (a1) -- (a6);
  
\end{tikzpicture}}
\caption{C.(b)-ii)-2} \label{Figure-0}
\end{figure}
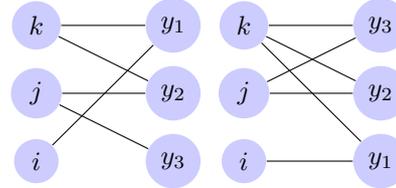

We have seen in our previous observations in \eqref{eq:A},\eqref{eq:B} and \eqref{eq:Ca} that $\min \limits_{\sigma \in \brr(\pi)} \bar{U}(\pi, \sigma)\leq 0, \forall \pi\in \AS \backslash {\rm C.(b)}$. Therefore  using \eqref{eq:supp-3}, \eqref{eq:supp-2} and \eqref{eq:supp-1} we conclude that 
\begin{equation}\label{eq:final_sup}
\sup \limits_{\pi\in \AS} \min \limits_{\sigma \in \brr(\pi)} \bar{U}(\pi, \sigma)=\frac{1}{3}
\end{equation}
and a S.E. strategy $\pi^*\in \AS$ exists if and only if $\min \limits_{\sigma \in \brr(\pi^*)} \bar{U}(\pi^*, \sigma)=\frac{1}{3}$. If $ \pi^* $ exists, then it must lie in one of the sub-classes of C.(b). Say the supremum value $\frac{1}{3}$ is attained in class C.(b)-i). Then using \eqref{eq:c_a_i}, we conclude that 
\begin{equation}\label{eq:contra1}
\begin{split}
 \min \limits_{\sigma \in \brr(\pi)} \bar{U}(\pi^*, \sigma) &= \frac{1}{3}[\pi^*(y_1|k) -\pi^*(y_2|k)]\ut(i,k) =\frac{1}{3},\\
\end{split}
\end{equation}
where $ \{i,j\} =\Dscr(\sigma,\pi^*) $ for all $ \sigma \in \brr(\pi^*) $ and $ k \neq i,j. $
Therefore, if $\ut(i,k)=1$, then we must have $\pi^*(y_2|k)=0,\pi^*(y_1|k)=1$  and if $\ut(i,k)=-1$ then we must have $\pi^*(y_2|k)=1,\pi^*(y_1|k)=0$ . But note that if $\pi^*(y_2|k)=0$ or $\pi^*(y_1|k)=0$ then $\pi^*$ belongs to class C.(a)-i), whereby \eqref{eq:contra1} contradicts \eqref{eq:cai}.  As a consequence, the supremum in \eqref{eq:final_sup} is never attained in C.(b)-i). 

Next assume that the supremum is attained in class C.(b)-ii)-1). Then there exists a S.E. strategy of the sender, $\pi^*$ in class C.(b)-ii)-1), such that
\begin{equation} \label{eq:contra2}
\begin{gathered}
\min \limits_{\sigma \in \brr(\pi^*)} \frac{1}{3}[-\pi^* (y_1|k) +\pi^* (y_2|k) + \\\pi^* (y_3|k)(-\sigma (i|y_3) + \sigma (j|y_3))]\ut(j,k) =\frac{1}{3}.\\
\end{gathered}
\end{equation}
If $\ut(j,k)=1$ then 
\begin{equation}\label{eq:final01}
\frac{1}{3}[-\pi^* (y_1|k) +\pi^* (y_2|k) -\pi^* (y_3|k)] =\frac{1}{3}.
\end{equation}
whereby $\pi^* (y_2|k)=1$ which gives $E_k=\{y_2\}$. On the other hand if $\ut(j,k)=-1$ then 
\begin{equation}\label{eq:final02}
\frac{1}{3}[\pi^* (y_1|k) -\pi^* (y_2|k) - \pi^* (y_3|k)] =\frac{1}{3}.
\end{equation}
whereby $\pi^* (y_1|k)=1$ which gives $E_k=\{y_1\}$. But if  $E_k=\{y_2\}$ or $E_k=\{y_1\}$ then $\pi^* \in $ C.(a)-i). But for such a $ \pi^* $ we have shown \eqref{eq:cai}, which contradicts \eqref{eq:contra2}. Similar arguments can be used to show that supremum is never attained in class C.(b)-ii)-2).

In summary, a Stackelberg equilibrium does not exist for this utility function of the sender. 
\end{examp}
This happens since the structure of the utility function is such that $[\ut_3]$ is a skew-symmetric matrix and $\ut(j,k)=-\ut(i,k),\forall i,j,k \in \Xscr$. Therefore, $\min \limits_{\sigma \in \brr(\pi) } \bar{U}(\pi,\sigma)$ contains $\pi(y|x)$ terms with positive and negative coefficients, where $x \notin \Dscr(\sigma,\pi)$. But making any $\pi(.|x)$ with  positive coefficient exactly equal to $1$, transforms the problem to one where $\min \limits_{\sigma \in \brr(\pi) } \bar{U}(\pi,\sigma)<0$. 
	Therefore, we can conclude that, unlike the deterministic case, the existence of an equilibrium in behavioral strategies is not guaranteed for every  utility function which proves Theorem \ref{theorem:existence_not_guranteed}. }
\section{Conclusion}\label{sec6}
Communication theory has concerned itself with measures of information content in a source or the capacity for transfer of information in a channel. However, these results have validity in the setting of cooperative communication where the sender and receiver have a common goal.
Our paper studied a setting where a strategic sender signals to a receiver to maximize its utility while the receiver seeks to know the true information known to the sender with zero error.  Our main intent was to arrive at and characterize a measure of how informative such an interaction could be for the receiver, when the receiver is the follower.

We introduced a naturally defined quantity called the \textit{informativeness} of  the sender which is the minimum number of source symbols that can be correctly recovered by the receiver in any equilibrium of a Stackelberg game in a deterministic setting. We characterized it and then demonstrated certain properties of the informativeness measure.

{\color{red}Additionally, we also looked at the setting where the players played randomized behavioural strategies. From the study we came to the conclusion that the existence of equilibrium is not guaranteed in this particular setting.}
Translating these results to the large blocklength regime as done in information theory and also in~\cite{vora2022shannon,vora2020zero,vora2020information} is a task for the future. 

\section*{Acknowledgments}
This research was supported by the grant CRG/2019/002975 of the Science and Engineering Research Board, Department of Science and Technology, India. The authors would also like to thank the two anonymous reviewers and the AE for their perceptive comments on an earlier version of this paper.
\bibliographystyle{unsrt}

\bibliography{new_ref.bib}

\end{document}